\documentclass{article}
\usepackage[
left=3cm,
right=3cm,
top=3cm,
bottom=3cm,
headsep=0.5cm,
footskip=0.75cm
]{geometry}
\usepackage[hidelinks]{hyperref}
\usepackage{amsmath}
\usepackage{parskip} 
\usepackage{mathrsfs}
\usepackage{pdfpages}
\usepackage{natbib}
\usepackage{bm}
\usepackage{booktabs}
\usepackage{eurosym} 
\usepackage{amsthm}
\newtheorem{theorem}{Theorem}[section]
\newtheorem{lemma}{Lemma}[section]
\newtheorem{proposition}{Proposition}[section]
\usepackage{lastpage} 
\usepackage{graphics}
\usepackage{multicol}
\usepackage[bf]{caption}
\usepackage{comment}
\usepackage{color}
\usepackage[utf8]{inputenc}

\title{\textbf{On inhomogeneity parameters for Backus average}}
\author{
Filip P. Adamus%
\footnote{
Department of Earth Sciences, Memorial University of Newfoundland, Canada, {\tt adamusfp@gmail.com}}\,,
Ayiaz Kaderali\footnote{%
Department of Earth Sciences, Memorial University of Newfoundland, Canada, {\tt ayiazkaderali@gmail.com}}
}

\date{}

\begin{document}

\maketitle
\begin{abstract}
In this paper, we discuss five parameters that indicate the inhomogeneity of a stack of parallel isotropic layers. 
We use field data to check their applicability.
We show that, in certain situations, they provide further insight into the intrinsic inhomogeneity of a Backus medium, as compared to the Thomsen parameters. 
Additionally, we show that the Backus average of isotropic layers is isotropic if and only if $\gamma=0$\,. 
This is in contrast to parameters $\delta$ and $\epsilon$, whose zero values do not imply isotropy.
\end{abstract}
\section{Introduction}
In this paper, we consider an inhomogeneous stack of thin, isotropic and parallel layers. We examine several parameters, which, by using the \citet{Backus1962} and \citet{Voigt1910} averages, indicate the strength of inhomogeneity. Using the former, we consider an inhomogeneous stack of thin isotropic layers, as a homogeneous transversely isotropic medium. In other words, the Backus average is a homogenization of inhomogeneity. The Voigt average represents an anisotropic medium, as the closest---in a Frobenius sense---isotropic counterpart. Among the parameters that we consider, we include the \citet{Thomsen1986} parameter $\gamma$. In addition to indicating anisotropy of the resulting transversely isotropic medium, $\gamma$ shows the inhomogeneity of the stack of layers. 
Specifically, we emphasize two parameters that refer to different methods of homogenization of isotropic layers to their isotropic counterparts. 
\section{Background}\label{sec:parameters}
\subsection{Backus and Voigt averages} \label{sec:BV}
According to \citet{Backus1962}, a sequence of thin parallel isotropic layers can be considered as a transversely isotropic medium. 
One of the few restrictions imposed by \citet{Backus1962} is that of long wavelengths and fine layering. 
The following elasticity parameters constitute the elasticity tensor that characterizes the medium resulting from the averaging process.
\begin{equation*}
   \label{bac1}
    c_{1111}^{\overline{\rm TI}}=\overline{\left(\frac{c_{1111}-2c_{2323}}{c_{1111}}\right)}^2\,\overline{\left(\frac{1}{c_{1111}}\right)}^{-1}+\overline{\left(\frac{4(c_{1111}-c_{2323})c_{2323}}{c_{1111}}\right)}\,,
\end{equation*}
\begin{equation*}
\label{bac3}
 c_{1133}^{\overline{\rm TI}}=\overline{\left(\frac{c_{1111}-2c_{2323}}{c_{1111}}\right)}\,\overline{\left(\frac{1}{c_{1111}}\right)}^{-1}\,,
 \end{equation*}
 \begin{equation}
    \label{bac4}
     c_{1212}^{\overline{\rm TI}}=\overline{c_{2323}}\,,
 \end{equation}
 \begin{equation}
     \label{bac5}
     c_{2323}^{\overline{\rm TI}}=\overline{\left(\frac{1}{c_{2323}}\right)}^{-1}\,,
 \end{equation}
 \begin{equation*}
 \label{bac6}
     c_{3333}^{\overline{\rm TI}}=\overline{\left(\frac{1}{c_{1111}}\right)}^{-1}\,,
 \end{equation*}
as shown by \citet{Backus1962} and discussed by \citet[Section 4.2]{Slawinski2020b}. 

The \citeauthor{Voigt1910} average results in an isotropic medium; its parameters are
\begin{equation*}
\label{voigt1}
    c_{1111}^{\rm iso}=\frac{1}{9}\left(2(c_{1111}+c_{2222}+c_{3333}+c_{1212}+c_{1313}+c_{2323})+c_{1122}+c_{1133}+c_{2233}\right)
\end{equation*}
and
\begin{equation*}
\label{voigt2}
    c_{2323}^{\rm iso}=\frac{1}{18}\left(c_{1111}+c_{2222}+c_{3333}+4(c_{1212}+c_{1313}+c_{2323})-(c_{1122}+c_{1133}+c_{2233})\right)\,,
\end{equation*}
where $c_{ijk\ell}$ refers to a generally anisotropic tensor.
We use Frobenius-21 norm,~$F_{21}$, because, according to \citet{DanekEtAl2015}, it lends itself to statistical analysis more easily than Frobenius-36 norm,~$F_{36}$.
Parameters $c^{\rm iso}_{1111}$ and $c^{\rm iso}_{2323}$ represent the closest isotropic tensor---in the $F_{21}$ sense---to the anisotropic one.

For a transversely isotropic tensor that results from the \citeauthor{Backus1962} average, the Voigt average results in
\begin{equation}
   \label{bv1}
     c_{1111}^{\rm iso}=\frac{1}{9}\left(5c_{1111}^{\overline{\rm TI}}+2c_{1133}^{\overline{\rm TI}}+4c_{2323}^{\overline{\rm TI}}+2c_{3333}^{\overline{\rm TI}}\right)
\end{equation}
and
\begin{equation}
   \label{bv2}
    c_{2323}^{\rm iso}=\frac{1}{18}\left(c_{1111}^{\overline{{\rm TI}}}-2c_{1133}^{\overline{{\rm TI}}}+6c_{1212}^{\overline{{\rm TI}}}+8c_{2323}^{\overline{{\rm TI}}}+c_{3333}^{\overline{{\rm TI}}}\right)\,.
\end{equation}
Herein, expressions (\ref{bv1}) and (\ref{bv2}) represent an isotropic counterpart to a stack of layers. 
To distinguish $c_{1111}^{\rm iso}$ and $c_{2323}^{\rm iso}$ from parameters obtained by arithmetic averaging, which are denoted below by $\overline{c_{1111}}$ and $\overline{c_{2323}}$, we let

\begin{equation*}
c_{1111}^{\rm iso}=:c^{\overline{\rm BV}}_{1111}
\end{equation*}
and
\begin{equation*}
c_{2323}^{\rm iso}=:c^{\overline{\rm BV}}_{2323}\,,
\end{equation*}
where $\overline{\rm BV}$ denotes the Backus-Voigt homogenization process.
In Section~\ref{sec:second}, we use $c^{\overline{\rm BV}}_{1111}$ and $c^{\overline{\rm BV}}_{2323}$ to define parameters measuring inhomogeneity of a stack of layers.
\subsection{Thomsen parameters}
To examine the strength of anisotropy of a transversely isotropic medium, we invoke \citeauthor{Thomsen1986} parameters
 \begin{equation}\label{eq:thom1}
     \gamma=\frac{c_{1212}^{\overline{\rm TI}}-c_{2323}^{\overline{\rm TI}}}{2c_{2323}^{\overline{\rm TI}}}\,,
 \end{equation}
\begin{equation*}\label{eq:thom2}
     \delta=\frac{\left(c_{1133}^{\overline{\rm TI}}+c_{2323}^{\overline{\rm TI}}\right)^2-\left(c_{3333}^{\overline{\rm TI}}-c_{2323}^{\overline{\rm TI}}\right)^2}{2c_{3333}^{\overline{\rm TI}}\left(c_{3333}^{\overline{\rm TI}}-c_{2323}^{\overline{\rm TI}}\right)}\,,
 \end{equation*}
\begin{equation*}\label{eq:thom3}
     \epsilon=\frac{c_{1111}^{\overline{\rm TI}}-c_{3333}^{\overline{\rm TI}}}{2c_{3333}^{\overline{\rm TI}}}\,.
 \end{equation*}
 As shown by~\citet{AdamusEtAl2018}, by increasing their values, these parameters indicate an increase of inhomogeneity of a stack of isotropic layers. 
 \subsection{Stability conditions}
Stability conditions \citep[e.g.,][Section 4.3]{Slawinski2020a} originate from the necessity of expending energy to deform a material. 
This necessity is mathematically expressed by the positive definiteness of the elasticity tensor. 
In general, a tensor is positive definite if and only if all its eigenvalues are positive. 
For an isotropic elasticity tensor, this entails that
\begin{equation}
\label{ineq}
    c_{1111}>\tfrac{4}{3}\,c_{2323}\,.
\end{equation}
According to~\citet{Backus1962}, a medium obtained by Backus averaging is positive definite if the layers, prior to averaging, are also positive definite. 
Also, according to~\citet{GazisEtAl1963}, a Frobenius-norm counterpart of a positive-definite tensor is positive definite. 
Thus, it suffices to ensure condition~(\ref{ineq}) for each layer.
\section{Parameters indicating inhomogeneity}
\subsection{Inhomogeneity parameters for Backus average}\label{sec:second}
In this paper, we consider five parameters that measure the inhomogeneity of a stack of isotropic layers. 
To obtain them, we use the averaging processes and expressions stated in Section~\ref{sec:BV}.
The Backus average allows us to relate wellbore information to seismic data.

As stated by \citet{Backus1962}, isotropic layers whose $c_{2323}$ is constant result in an isotropic Backus medium. To examine the inhomogeneity of such layers, we introduce 
\begin{equation}\label{eq:Inh1}
    \mathscr{I}:=\frac{\overline{c_{1111}}-c^{\overline{\rm TI}}_{3333}}{2c^{\overline{\rm TI}}_{3333}}
\end{equation}
and
\begin{equation*}
    \mathscr{I}_{BV}:=\frac{\overline{c_{1111}}-c^{\overline{\rm BV}}_{1111}}{2c^{\overline{\rm BV}}_{1111}}\,.
\end{equation*}
Equation~(\ref{eq:Inh1}) relates the elasticity parameters of the layers to those of a transversely isotropic medium resulting from the Backus average. 
For an isotropic medium, $c^{\overline{\rm TI}}_{3333}=c^{\overline{\rm TI}}_{1111}$.
Thus, $\mathscr{I}$ indicates only the differences among $c_{1111}$ within the stack of layers, as compared to $\mathscr{I}_{BV}$, which provides more complex information about inhomogeneity, since $c^{\overline{\rm BV}}_{1111}$ depends on both  $c_{1111}$ and $c_{2323}$. 
$\mathscr{I}_{BV}$ shows the difference between two methods of homogenization of an inhomogeneous stack of isotropic layers to its isotropic counterpart.
In the inverse problem---where we only know Backus parameters provided by seismic information---$\mathscr{I}$ and $\mathscr{I}_{BV}$ cannot be used.

Another two parameters to measure inhomogeneity are
\begin{equation*}
    \gamma=\frac{\overline{c_{2323}}-c_{2323}^{\overline{\rm TI}}}{2c_{2323}^{\overline{\rm TI}}}= \frac{c_{1212}^{\overline{\rm TI}}-c_{2323}^{\overline{\rm TI}}}{2c_{2323}^{\overline{\rm TI}}}
\end{equation*}
and
\begin{equation*}
    \gamma_{BV}:=\frac{\overline{c_{2323}}-c_{2323}^{\overline{\rm BV}}}{2c_{2323}^{\overline{\rm BV}}}= \frac{c_{1212}^{\overline{\rm TI}}-c_{2323}^{\overline{\rm BV}}}{2c_{2323}^{\overline{\rm BV}}}\,,
\end{equation*}

where $\gamma$ is parameter~(\ref{eq:thom1}).
As shown in Theorem~\ref{th:one} in Appendix~\ref{sec:appa}, the Backus average of isotropic layers is isotropic if and only if $\gamma=0$\,, in contrast to parameters $\delta$ and $\epsilon$, whose zero values do not imply isotropy. 
Thus, in this paper, we do not use  $\delta$ and $\epsilon$. 

Parameter $\gamma_{BV}$---in comparison to $\gamma$---gives different information about inhomogeneity of a stack, since $c^{\overline{\rm BV}}_{2323}$ depends on both parameters $c_{1111}$ and $c_{2323}$, not only on $c_{2323}$. 
Also, $\gamma_{BV}$ distinguishes between two methods of homogenization, whereas $\gamma$ does not. 

The last parameter we use is
\begin{equation*}
    \mathscr{N}:=||C_{\rm \overline{TI}}||_{F_{21}}-||C_{\rm \overline{BV}}||_{F_{21}}\,,
\end{equation*}
which indicates the difference between the transversely isotropic tensor resulting from Backus average and the effective isotropic tensor resulting from the Backus-Voigt average, where
\begin{equation*}
    ||C_{\rm \overline{TI}}||_{F_{21}}=\left(2\left(c^{\overline{\rm TI}}_{1111}\right)^{\!2}+4\left(c_{1133}^{\overline{\rm TI}}\right)^{\!2}+2\left(c_{1111}^{\overline{\rm TI}}-2c_{1212}^{\overline{\rm TI}}\right)^{\!2}+\left(c_{3333}^{\overline{\rm TI}}\right)^{\!2}+2\left(2c_{2323}^{\overline{\rm TI}}\right)^{\!2}+\left(2c_{1212}^{\overline{\rm TI}}\right)^{\!2}\right)^{\!\frac{1}{2}}\,
\end{equation*}
and 
\begin{equation*}
||C_{\rm \overline{BV}}||_{F_{21}} =
\left(
3\left(c_{1111}^{\overline{\rm BV}}\right)^{\!2} +
6\left(c_{1111}^{\overline{\rm BV}}-2c_{2323}^{\overline{\rm BV}}\right)^{\!2} +
3\left(2c_{2323}^{\overline{\rm BV}}\right)^{\!2}
\right)^{\!\frac{1}{2}}\,.
\end{equation*}
Herein, $\mathscr{N}$ indicates inhomogeneity of a stack of layers, as well as the anisotropy of the medium.
Similarly to the Voigt average, we use the $F_{21}$ norm.
\subsection{Constant rigidity: Isotropic medium}
To illustrate parameters $\mathscr{I}$ and $\mathscr{I}_{BV}$, let us consider a stack of isotropic layers with elasticity parameters shown in Table~\ref{tab:appone1}.
\begin{table}[h]
    \centering
    \begin{tabular}{ ||p{2cm}||p{2cm}||}
\hline
 $c_{1111}$ & $c_{2323}$ \\
 \hline\hline
 10$x$&2\\
 \hline
 10&2\\
 \hline
 10$x$&2\\
 \hline
 10&2\\\hline
 10$x$&2\\
 \hline
 10&2\\\hline
 10$x$&2\\
 \hline
 10&2\\\hline
 10$x$&2\\
 \hline
 10&2\\
 \hline
\end{tabular}
    \caption{\small{Elasticity parameters for ten isotropic layers; factor $x$ controls the inhomogeneity of the stack.}}
    \label{tab:appone1}
\end{table}

Figure \ref{fig:appone1}---for $x=1$---represents homogeneous stack, where $c_{1111}=10$ and $c_{2323}=2$. 
As $x$ increases, the inhomogeneity of $c_{1111}$ increases. 

\begin{figure}[h]
    \centering
    \includegraphics[scale=0.5]{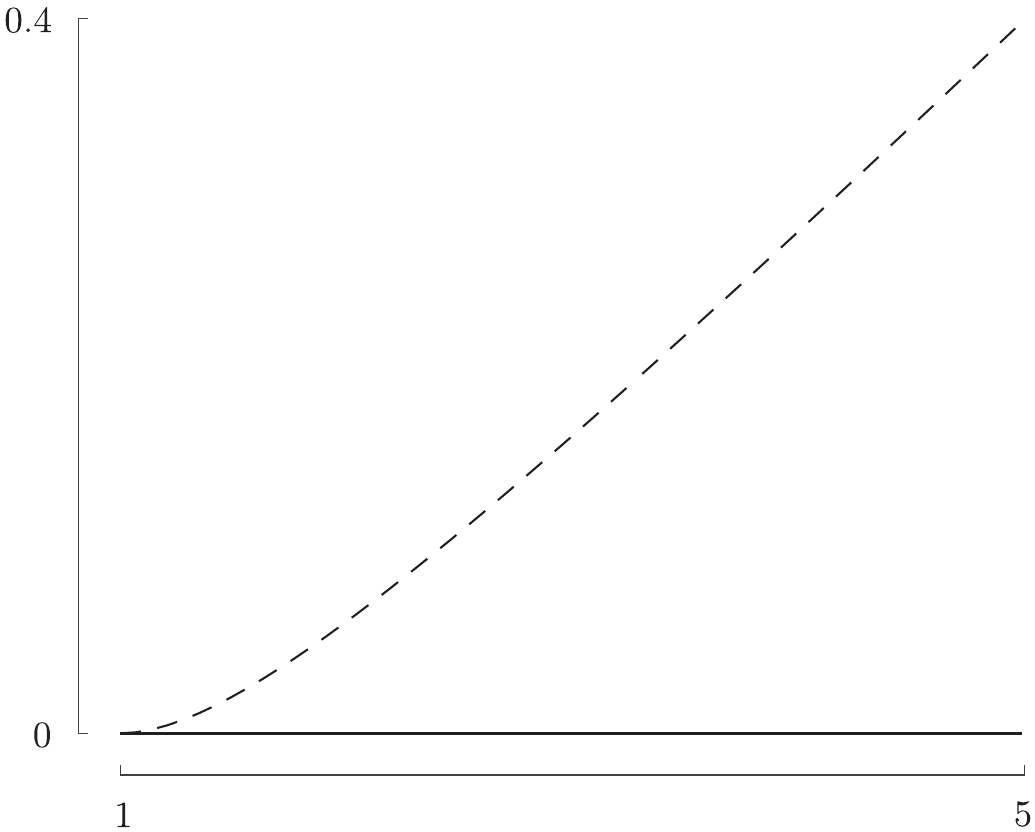}
    \caption{\small{Horizontal axis exhibits values of $x$. Values of $\mathscr{I}$ are shown by a dashed black line, $\mathscr{I}_{BV}$ by a dashed grey line, $\gamma$ by a dotted black line, $\gamma_{BV}$ by a dotted grey line and  $\mathscr{N}$ by a solid black line; $\gamma$, $\gamma_{BV}$ and $\mathscr{N}$ are overlain at zero value of vertical axis. Also, $\mathscr{I}$ and $\mathscr{I}_{BV}$ coincide.}}
    \label{fig:appone1}
\end{figure}

Only $\mathscr{I}$ and $\mathscr{I}_{BV}$ indicate growing inhomogeneity of $c_{1111}$. 
The other parameters are zero; they indicate no inhomogeneity and no anisotropy. 
$\mathscr{I}$ and $\mathscr{I}_{BV}$ are equal to each other, because, for isotropy, $c^{\overline{\rm TI}}_{3333}=c^{\overline{\rm BV}}_{1111}$. 
For the case of constant rigidity, the medium is isotropic, as a consequence $\gamma=0$; herein, Thomsen parameters $\delta$ and $\epsilon$ are also zero. 
\newpage
\subsection{Near constant rigidity: Anisotropic medium} \label{sec:quasi}
Let us consider an example of a non--significantly varying $c_{2323}$. 
\begin{table}[h]
    \centering
    \begin{tabular}{ ||p{2cm}||p{2cm}||}
\hline
 $c_{1111}$ & $c_{2323}$ \\
 \hline\hline
 10$x$&3\\
 \hline
 10&2\\
 \hline
 10$x$&3\\
 \hline
 10&2\\\hline
 10$x$&3\\
 \hline
 10&2\\\hline
 10$x$&3\\
 \hline
 10&2\\\hline
 10$x$&3\\
 \hline
 10&2\\
 \hline
\end{tabular}
    \caption{\small{Elasticity parameters for ten isotropic layers; factor $x$ controls the inhomogeneity of the stack.}}
    \label{tab:appquasi}
\end{table}

\begin{figure}[h]
    \centering
    \includegraphics[scale=0.5]{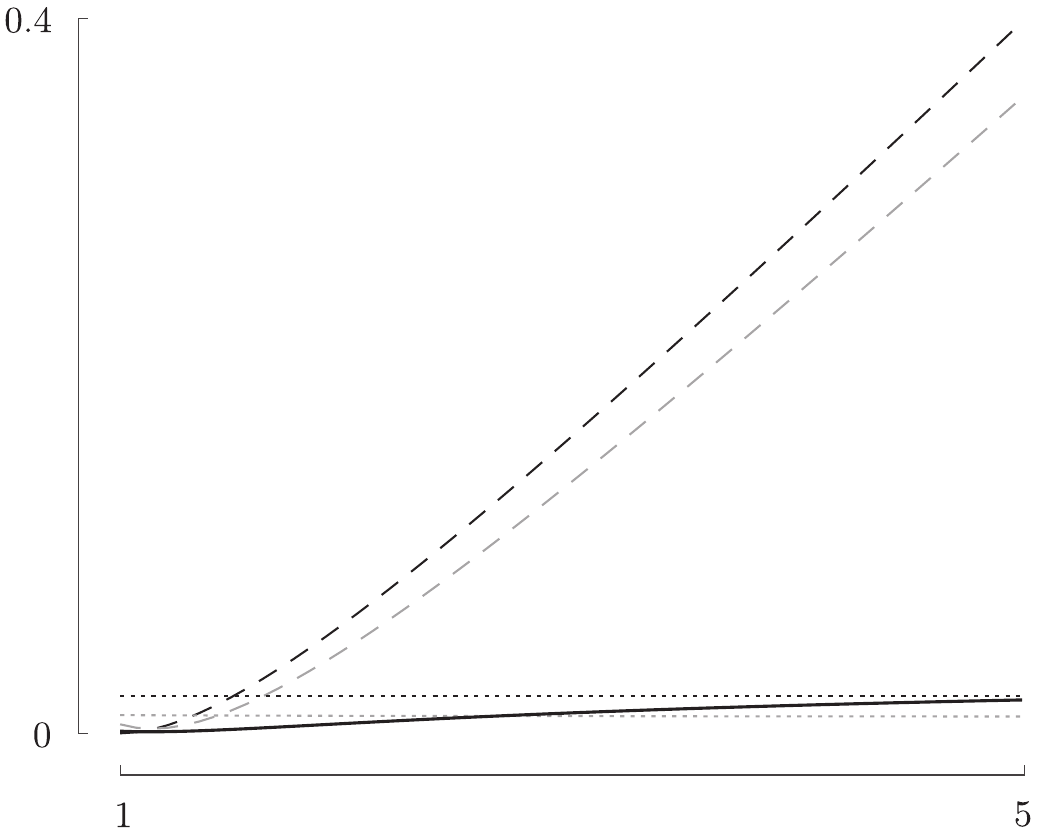}
    \caption{\small{Horizontal axis exhibits values of $x$. Values of $\mathscr{I}$ are shown by a dashed black line, $\mathscr{I}_{BV}$ by a dashed grey line, $\gamma$ by a dotted black line, $\gamma_{BV}$ by a dotted grey line and $\mathscr{N}$ by a solid black line.}}
    \label{fig:appquasi}
\end{figure}

As shown on Figure~\ref{fig:appquasi}, in general, $\mathscr{I}$ exhibits larger values than $\mathscr{I}_{BV}$. 
This stems from the exclusive dependance of inhomogeneity of $c_{1111}$ for $\mathscr{I}$.
However, for very low values of~$x$---where inhomogeneity of $c_{1111}$ is weaker than that of $c_{2323}$---$\mathscr{I}$ has lower values than $\mathscr{I}_{BV}$.
This results from the dependance of inhomogeneity of $c_{2323}$ for $\mathscr{I}_{BV}$. 
$\gamma$ is approximately twice as large as $\gamma_{BV}$; the inhomogeneity of $c_{1111}$ does not influence parameter $\gamma$ and, for the case of low inhomogeneity of $c_{2323}$, has a negligible effect on $\gamma_{BV}$, due to the nature of equation~(\ref{bv2}).
$\mathscr{N}$ represents the inhomogeneity of $c_{1111}$ and $c_{2323}$, as expected.
\subsection{Equally--scaled elasticity parameters: Anisotropic medium}\label{sec:scal}
Let us consider an example to illustrate that every parameter indicates inhomogeneity, and to exhibit the relationship between them.
In Table~\ref{tab:apptwo}, the inhomogeneity grows equally for both elasticity parameters; Figure~\ref{fig:apptwo} represents such a situation.
\begin{table}[h]
    \centering
    \begin{tabular}{ ||p{2cm}||p{2cm}||}
\hline
 $c_{1111}$ & $c_{2323}$ \\
 \hline\hline
 10$x$&2$x$\\
 \hline
 10&2\\
 \hline
 10$x$&2$x$\\
 \hline
 10&2\\\hline
 10$x$&2$x$\\
 \hline
 10&2\\\hline
 10$x$&2$x$\\
 \hline
 10&2\\\hline
 10$x$&2$x$\\
 \hline
 10&2\\
 \hline
\end{tabular}
    \caption{\small{Elasticity parameters for ten isotropic layers; factor $x$ controls the inhomogeneity of the stack.}}
    \label{tab:apptwo}
\end{table}
\begin{figure}[h]
    \centering
    \includegraphics[scale=0.5]{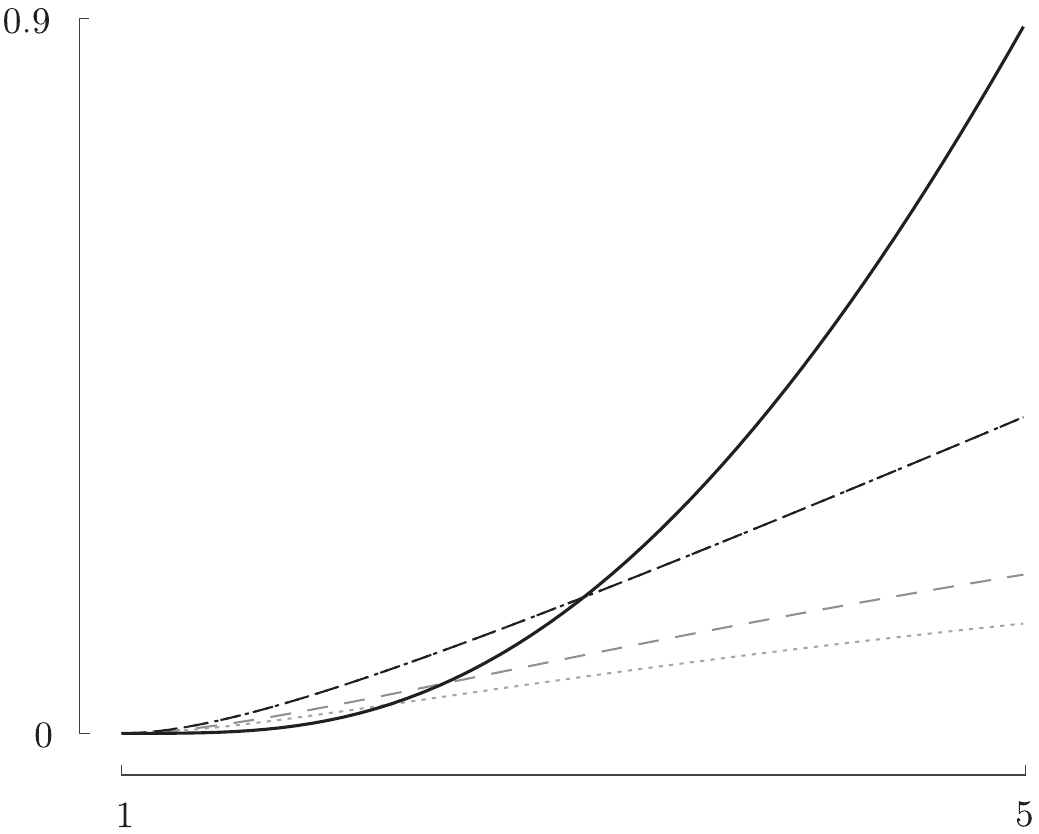}
    \caption{\small{Horizontal axis exhibits values of $x$. Values of $\mathscr{I}$ are shown by a dashed black line, $\mathscr{I}_{BV}$ by a dashed grey line, $\gamma$ by a dotted black line, $\gamma_{BV}$ by a dotted grey line and $\mathscr{N}$ by a solid black line; $\mathscr{I}$ and $\gamma$ are overlain.}}
    \label{fig:apptwo}
\end{figure}

For weak inhomogeneity, all five parameters have similar values.
Also, $\mathscr{I}$ and $\gamma$ have the same values for strong inhomogeneity. 
This comes from the fact that, in this example, the inhomogeneity of $c_{1111}$ and $c_{2323}$ grows proportionally, and $\mathscr{I}$ indicates only inhomogeneity of $c_{1111}$ while $\gamma$ of $c_{2323}$.
Thus, we conclude that for similar inhomogeneity of $c_{1111}$ and $c_{2323}$, $\mathscr{I}$ and $\gamma$ have similar values.
For strong inhomogeneity, $\mathscr{N}$ has much larger values than $\mathscr{I}$, $\mathscr{I}_{BV}$, $\gamma$ and $\gamma_{BV}$.
Comparing Figures~\ref{fig:appquasi} and \ref{fig:apptwo}, we conclude that $\mathscr{N}$ is more sensitive to the inhomogeneity of $c_{2323}$ as opposed to that of $c_{1111}$.
As the value of $x$ increases, the difference between $\gamma$ and $\gamma_{BV}$ also increases.
For $x=5$, $\gamma$ is approximately three times as large as $\gamma_{BV}$. 
Hence, a large difference between $\gamma$ and $\gamma_{BV}$ indicates strong inhomogeneity of $c_{2323}$ and---as shown in a similar example in Appendix~\ref{sec:appb}---strong inhomogeneity of $c_{1111}$.
\subsection{Real data: Anisotropic medium} \label{sec:real}
Let us consider an example using well logging measurements offshore Eastern Canada~\citep{ZhouKaderali2006}.
We segment the 2.4-kilometre region into seven layers.
Each layer exhibits varying inhomogeneity in $c_{1111}$ and $c_{2323}$\,, which results in varying amounts of anisotropy in the equivalent TI layers resulting from the Backus average.
We tabulate the Backus parameters in Table~\ref{tab:borehole}, and use them to calculate the layer values of $\mathscr{I}$\,, $\mathscr{I}_{BV}$\,, $\gamma$\,, $\gamma_{BV}$\,, and $\mathscr{N}$\,, which we display in Figure~\ref{fig:borehole}.

\renewcommand{\arraystretch}{1.3}
\begin{table}[h]
    \centering
    \begin{tabular}{ ||p{1cm}||p{1cm}|p{1cm}||p{1cm}|p{1cm}|p{1cm}|p{1cm}|p{1cm}||}
\hline
 Layers & $\overline{c_{1111}}$ &$\overline{c_{2323}}$&$c_{1111}^{\overline{\rm TI}}$&$c_{1133}^{\overline{\rm TI}}$&$c_{1212}^{\overline{\rm TI}}$&$c_{2323}^{\overline{\rm TI}}$&$c_{3333}^{\overline{\rm TI}}$ \\
 \hline\hline
 1&7.21&1.69&7.20&3.81&1.69&1.66&7.17\\
 \hline
 2&9.55&2.65&9.53&4.21&2.65&2.52&9.35\\
 \hline
 3&14.77&4.33&14.74&6.06&4.33&4.26&14.61  \\
 \hline
 4&14.68&4.75&14.64&5.18&4.75&4.62&14.54\\\hline
 5&15.19&5.70&15.16&3.79&5.70&5.68&15.18\\
 \hline
 6&15.94&6.64&15.83&2.50&6.64&6.57&15.60\\\hline
 7&15.70&5.40&15.70&4.91&5.40&5.36&15.66\\
 \hline
\end{tabular}
    \caption{\small{Arithmetic averaging of $c_{1111}$ and $c_{2323}$, and Backus parameters for seven thick layers composed of many thin layers. Units of parameters are $[{\rm m}^2/{\rm s}^2]\times10^6$.}}
    \label{tab:borehole}
\end{table}

\begin{figure}[h]
    \centering
    \includegraphics[scale=0.3]{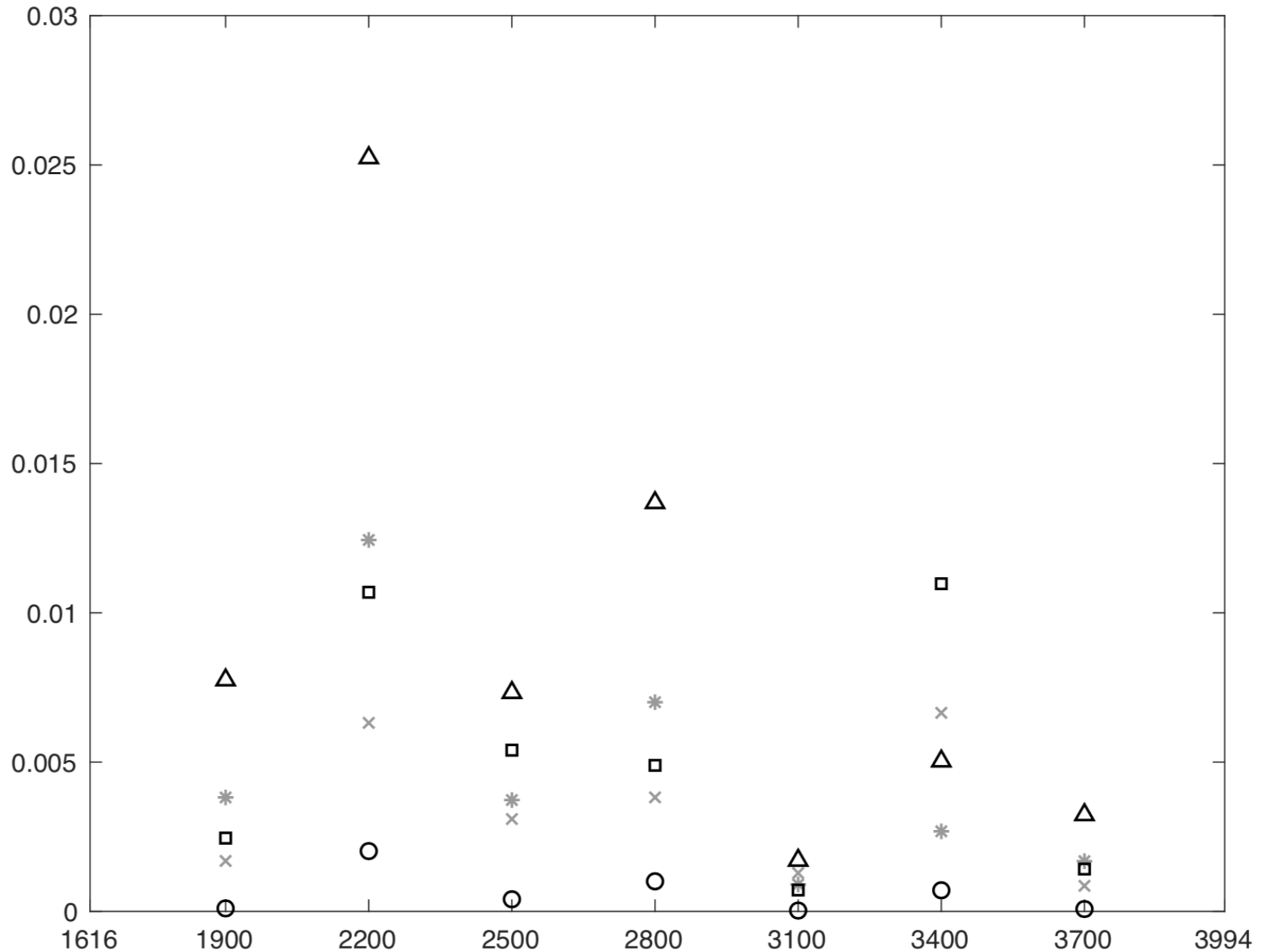}
    \caption{\small{Bottom axis exhibits depth, values of parameter $\mathscr{I}$ is shown by black squares, $\mathscr{I}_{BV}$ by grey crosses, $\gamma$ by black triangles, $\gamma_{BV}$ by grey stars and  $\mathscr{N}$ by black circles.}}
    \label{fig:borehole}
\end{figure}

As indicated by the values of the vertical axis in Figure~\ref{fig:borehole}, the region considered exhibits low anisotropy and inhomogeneity, especially in the fifth layer.
Throughout the entire region, $\mathscr{N}<\gamma$\,, which, referring to Section~\ref{sec:quasi}, is the result of weak inhomogeneity in both $c_{1111}$ and $c_{2323}$\,.
Since $\mathscr{I}$ indicates only inhomogeneity in $c_{1111}$\,, while $\gamma$ only of $c_{2323}$\,, we deduce that---apart from the sixth layer---the inhomogeneity of $c_{2323}$ is stronger than that of $c_{1111}$\,.
In the sixth layer, $\mathscr{I}$ and $\mathscr{I}_{BV}$ have larger values than $\gamma$\,, which indicates that the layer is more inhomogeneous in both $c_{1111}$ and $c_{2323}$\,; this is confirmed by $\mathscr{N}$\,, which is less sensitive to $c_{1111}$ than $c_{2323}$ and has a value less than $\gamma$\,.
Furthermore, the values of $\gamma$ are approximately twice as large as the value of $\gamma_{BV}$ in all layers, which indicates that the inhomogeneity of $c_{2323}$ is large.
However, since $\gamma$ and $\mathscr{I}$ are not very similar, we deduce that the inhomogeneity between each $c_{1111}$ and $c_{2323}$ is relatively different.

\subsection{Real data: Near constant--rigidity case}
\label{sec:NearConstantRigidity}
Let us consider a portion of layer one, whose elasticity parameters indicate a near-constant rigidity of layers.
We extract eight well log measurements that correspond to a 1.2-metre interval and tabulate their density-scaled elasticity parameters in Table~\ref{tab:borehole2a}.
We perform the Backus average on the interval and tabulate the Backus parameters, along with the relevant anisotropy parameters, in Table~\ref{tab:borehole2}.

\renewcommand{\arraystretch}{1.2}
\begin{table}[h]
    \centering
    \begin{tabular}{ ||p{2cm}||p{2cm}||}
\hline
\multicolumn{2}{||c||}{$[{\rm m}^2/{\rm s}^2]\times10^{6}$}\\
\hline \hline
 $c_{1111}$ & $c_{2323}$ \\
 \hline\hline
 7.74033	& 1.83139 \\ \hline
7.84884	& 1.83647 \\ \hline
8.12686	& 1.84012 \\ \hline
8.46952	& 1.84051 \\ \hline
8.83099	& 1.84097 \\ \hline
9.10564	& 1.83502 \\ \hline
9.12784	& 1.82696 \\ \hline
8.88130	& 1.82287 \\ \hline
\end{tabular}
    \caption{\small{Values exhibit isotropic elasticity parameters for each of eight thin layers. Parameter $c_{1111}$ varies significantly, while $c_{2323}$ remains nearly constant.}}
    \label{tab:borehole2a}
\end{table}
\renewcommand{\arraystretch}{1.3}
\begin{table}[h]
\centering
\begin{tabular}{||p{1.8cm}||*{2}{p{1.4cm}||}} 
\hline
\multicolumn{2}{||c||}{$[{\rm m}^2/{\rm s}^2]\times10^{6}$}\\
 \hline\hline
$c^{\overline{\rm TI}}_{1111}$&8.48373         \\ \hline
$c^{\overline{\rm TI}}_{1133}$&4.81539                \\ \hline
$c^{\overline{\rm TI}}_{1212}$&1.83429              \\ \hline
$c^{\overline{\rm TI}}_{2323}$&1.83427                    \\ \hline
$c^{\overline{\rm TI}}_{3333}$&8.48419                    \\ \hline
\end{tabular}
\quad
\begin{tabular}{||p{2.6cm}||*{2}{p{1.4cm}||}} 
\hline 
\multicolumn{2}{||c||}{$\times10^{-6}$}\\ \hline \hline
$\mathscr{I}$&1899.34\\
\hline
$\mathscr{I}_{BV}$&1917.93\\
\hline
$\gamma$&5.862
\\
\hline
$\gamma_{BV}$&2.673
\\
\hline
$\mathscr{N}\, \small{[{\rm m}^2/{\rm s}^2]\times10^6}$&0.005
\\
\hline
\end{tabular}
 \caption{\small{Values exhibit Backus elasticity parameters, $\mathscr{I}$, $\mathscr{I}_{BV}$, $\gamma$, $\gamma_{BV}$ and $\mathscr{N}$. }}
    \label{tab:borehole2}
\end{table}

Thomsen parameter $\gamma$ confirms that the layer is nearly isotropic, which is a consequence of the nearly constant elasticity parameter $c_{2323}$\,.
The values of $\mathscr{I}$ and $\mathscr{I}_{BV}$ are large in comparison to the other parameters, which indicates inhomogeneity of $c_{1111}$ between the layers.
However, $\mathscr{I}_{BV}>\mathscr{I}$\,, and thus, the inhomogeneity is not very significant, which is in agreement with Figure~\ref{fig:appquasi} for low values of scale factor $x$\,.
Within this layer, the value of $\gamma$ is nearly twice as large as the value of $\gamma_{BV}$\,, which---as indicated by Figures~\ref{fig:appquasi} and~\ref{fig:borehole}---is characteristic of a near constant-rigidity case.
Nevertheless, since the value of $\mathscr{N}$ is very low, we determine that the inhomogeneity of $c_{1111}$ and $c_{2323}$ is very low and, as such, the Backus average results in a nearly isotropic layer.

\section{Conclusions}
The five parameters stated in Section~\ref{sec:second} allow us to examine the inhomogeneity of a stack of layers resulting in a Backus medium.
In the case of isotropic layers with constant $c_{2323}$, we require $\mathscr{I}$ or $\mathscr{I}_{BV}$ to measure inhomogeneity using the Backus average. 
In this special case, the resulting medium is isotropic; hence the Thomsen parameters are equal to zero and they do not indicate the intrinsic inhomogeneity of a Backus medium.

$\mathscr{N}$ appears to be particularly useful in measuring inhomogeneity as it relies on both $c_{1111}$ and $c_{2323}$.
By combining the properties of three Thomsen parameters, it shows complex inhomogeneity.
It can be used in the inverse problem---where we only know the Backus parameters provided by seismic information---the same way as $\gamma$ and $\gamma_{BV}$.

Also, the relationship between $\gamma$ and $\gamma_{BV}$ indicates the inhomogeneity of $c_{2323}$, alongside the minor auxiliary influence of the inhomogeneity of $c_{1111}$.
For the case of near-constant rigidity, the relationship is approximately 2:1; the influence of $c_{1111}$ on this relationship is very small.
Stronger inhomogeneity of $c_{2323}$ affects this relationship.
In such a case, the influence of the inhomogeneity of $c_{1111}$ also increases; the relationship can reach 3:1 or more.

The relationship between $\mathscr{I}$ and $\mathscr{I}_{BV}$ may be insightful.
Larger values of $\mathscr{I}$ are characteristic for strong inhomogeneity of $c_{1111}$.
The case, where $\mathscr{I}_{BV}$ is larger, indicates low inhomogeneity of $c_{1111}$ and stronger influence of $c_{2323}$.

Similar values of parameters $\mathscr{I}$ and $\gamma$ indicate the case of similarly scaled $c_{1111}$ and $c_{2323}$, as shown in Section~\ref{sec:scal}.

In summary, the five parameters may be used to show the inhomogeneity, beyond Thomsen parameters, especially in the case of near-constant rigidity.
As is exemplified in Sections~\ref{sec:real} and~\ref{sec:NearConstantRigidity}, where real data are used, the five parameters allow us to perform a detailed analysis of the region. 
They indicate the seven layer medium to be of low anisotropy and inhomogeneity and illustrate the Backus average of the thin stack of layers to be nearly isotropic.

\section*{Acknowledgements}
We wish to acknowledge the supervision of Michael A. Slawinski, the proofreading and editing of Theodore Stanoev and the proofreading and graphic support of Elena Patarini.
This research was performed in the context of The Geomechanics Project supported by Husky Energy. 
Also, this research was partially supported by the Natural Sciences and Engineering Research Council of Canada, grant 238416-2013.

\bibliographystyle{apa}
\bibliography{AK}
\begin{appendix}
\section{Theorem A.1} \label{sec:appa}
\begin{theorem}\label{th:one}
The Backus average of isotropic layers is isotropic if and only if its $\gamma=0$\,, in contrast to parameters $\delta$ and $\epsilon$, whose zero values do not imply isotropy.
\end{theorem}
\begin{proof}
	\begin{lemma}
	If $\gamma=0$\,, then the Backus average of isotropic layers is isotropic.
	\end{lemma}
	\begin{proof}
	If $\gamma=0$, then $c_{1212}^{\overline{\rm TI}}=c_{2323}^{\overline{\rm TI}}$, and in accordance with expressions~(\ref{bac4}) and (\ref{bac5}),
	\begin{equation*}
	c_{1212}^{\overline{\rm TI}}=\overline{c_{2323}}\,
	\end{equation*}
	and
	\begin{equation*}
	c_{2323}^{\overline{\rm TI}}=\overline{\left(\frac{1}{c_{2323}}\right)}^{-1}\,,
	\end{equation*}
	we obtain,
	\begin{equation}\label{eq:last}
	\overline{c_{2323}}=\overline{\left(\frac{1}{c_{2323}}\right)}^{-1}\,.
	\end{equation}
	Equation~(\ref{eq:last}) is true, if and only if, $c_{2323}$ is constant.
	As stated by \citet{Backus1962} and discussed by \citet{AdamusEtAl2018}, layers whose $c_{2323}$ is constant result in an isotropic \citeauthor{Backus1962} medium.
	
	If Backus average is isotropic, then $c_{1212}^{\overline{\rm TI}}=c_{2323}^{\overline{\rm TI}}$, and, hence, $\gamma=0$\,.
	\end{proof}
	\begin{lemma}\label{lem:delta}
	If $\delta=0$\,, it does not follow that the Backus average of isotropic layers is isotropic.
	\end{lemma}
	\begin{proof}
	If $\delta=0$\,, then 
	\begin{equation}\label{eq:tue}
	\left(c_{1133}^{\overline{\rm TI}}+c_{2323}^{\overline{\rm TI}}\right)^2-\left(c_{3333}^{\overline{\rm TI}}-c_{2323}^{\overline{\rm TI}}\right)^2=\left(c_{1133}^{\overline{\rm TI}}\right)^2-\left(c_{3333}^{\overline{\rm TI}}\right)^2+2c_{2323}^{\overline{\rm TI}}\left(c_{1133}^{\overline{\rm TI}}+c_{3333}^{\overline{\rm TI}}\right)=0\,.
	\end{equation}
	Let us consider anisotropic Backus medium (from Proposition~\ref{prop:prop}); $c_{3333}^{\overline{\rm TI}}=2c_{1133}^{\overline{\rm TI}}=4c_{2323}^{\overline{\rm TI}}$. 
        In such a case, equation~(\ref{eq:tue}) becomes
	\begin{equation*}
	4\left(c_{2323}^{\overline{\rm TI}}\right)^2-16\left(c_{2323}^{\overline{\rm TI}}\right)^2+4\left(c_{2323}^{\overline{\rm TI}}\right)^2+8\left(c_{2323}^{\overline{\rm TI}}\right)^2=0\,,
	\end{equation*}
	which remains true for an anisotropic Backus average.
	\end{proof}
	\begin{lemma}\label{lem:eps}
	If $\epsilon=0$\,, it does not follow that the Backus average of isotropic layers is isotropic.
	\end{lemma}
	\begin{proof}
	If $\epsilon=0$\,, then 
	\begin{equation}\label{eq:monday2}
	c_{1111}^{\overline{\rm TI}}=c_{3333}^{\overline{\rm TI}}\,.
	\end{equation}
	Let us consider an anisotropic Backus, where $c_{2323}^{\overline{\rm TI}}\neq c_{1212}^{\overline{\rm TI}}$, $c_{1133}^{\overline{\rm TI}}\neq c_{1111}^{\overline{\rm TI}}-2c_{2323}^{\overline{\rm TI}}$ and $c_{1111}^{\overline{\rm TI}}=c_{3333}^{\overline{\rm TI}}$.
	Equation~(\ref{eq:monday2}) becomes
	\begin{equation*}
	c_{3333}^{\overline{\rm TI}}=c_{3333}^{\overline{\rm TI}}\,,
	\end{equation*}
	which remains true for an anisotropic Backus average.
	\end{proof}
	\begin{proposition}\label{prop:prop}
	A transversely isotropic tensor---with $c_{3333}^{\overline{\rm TI}}=2c_{1133}^{\overline{\rm TI}}=4c_{2323}^{\overline{\rm TI}}$---remains transversely isotropic.	
	\end{proposition}
	\begin{proof}
	Consider
\begin{equation*}
C=
\left[
\begin{array}{cccccc}
c_{1111} & c_{1111}-2c_{1212} & 2c_{2323} & 0 & 0 & 0\\
c_{1111}-2c_{1212} & c_{1111} & 2c_{2323} & 0 & 0 & 0\\
2c_{2323} & 2c_{2323} & 4c_{2323} & 0 & 0 & 0\\
0 & 0 & 0 & 2c_{2323} & 0 & 0\\
0 & 0 & 0 & 0 & 2c_{2323} & 0\\
0 & 0 & 0 & 0 & 0 & 2c_{1212}\end{array}
\right]
\,.
\end{equation*}
Its eigenvalues are
\begin{equation*}
\lambda_1=c_{1111}-c_{1212}+2c_{2323}-\sqrt{c_{1111}^2-4c_{1111}c_{2323}-2c_{1111}c_{1212}+12c_{2323}^2+4c_{2323}c_{1212}+c_{1212}^2}\,,
\end{equation*}
\begin{equation*}
\lambda_2=c_{1111}-c_{1212}+2c_{2323}+\sqrt{c_{1111}^2-4c_{1111}c_{2323}-2c_{1111}c_{1212}+12c_{2323}^2+4c_{2323}c_{1212}+c_{1212}^2}\,,
\end{equation*}
\begin{equation*}
\lambda_3=\lambda_4=2c_{2323}\,,
\end{equation*}
\begin{equation*}
\lambda_5=\lambda_6=2c_{1212}\,.
\end{equation*}
The eigenvalue multiplicities and their corresponding spaces of eigentensors---according to the Theorem 4.3 of~\citet{BonaEtAl2007}---imply that $C$ is a transversely isotropic tensor, as required.
	\end{proof}
	This completes the proof.
\end{proof}

To numerically exemplify Lemma~\ref{lem:delta}, let us consider a stack of two isotropic layers; herein, the parameters for the first layer are $c_{1111}=4$ and $c_{2323}=1$, while for the second $c_{1111}=1$ and $c_{2323}=0.250$. 
For such a case, the Backus average---wherein $c_{1111}^{\overline{\rm TI}}=2.275$, $c_{1133}^{\overline{\rm TI}}=0.800$, $c_{1212}^{\overline{\rm TI}}=0.625$, $c_{2323}^{\overline{\rm TI}}=0.400$ and $c_{3333}^{\overline{\rm TI}}=1.600$---is not isotropic. 
Specifically, $\delta=0$, $\epsilon=0.211$ and $\gamma=0.281$.

To numerically exemplify Lemma~\ref{lem:eps}, let us consider a stack of two isotropic layers; herein, the parameters for the first layer are $c_{1111}=2$ and $c_{2323}=1$, while for the second $c_{1111}=1.200$ and $c_{2323}=0.200$. 
For such a case, the Backus average---wherein $c_{1111}^{\overline{\rm TI}}=1.500$, $c_{1133}^{\overline{\rm TI}}=0.500$, $c_{1212}^{\overline{\rm TI}}=0.600$, $c_{2323}^{\overline{\rm TI}}=0.333$ and $c_{3333}^{\overline{\rm TI}}=1.500$---is not isotropic.
Specifically, $\epsilon=0$, $\delta=-0.190$ and $\gamma=0.400$.
\newline
\begin{proposition}
The Backus average of isotropic layers is isotropic if and only if its $\delta=0$ and $\epsilon=0$.
\end{proposition}
\begin{proof}
Let us consider a stack of two isotropic layers. 
We denote elasticity parameters for the first layer as $c_{1111}=a$ and $c_{2323}=c$, and for the second as $c_{1111}=b$ and $c_{2323}=d$.
For the Backus average, $\delta=0$ if and only if 
\begin{equation}\label{eq:deltaprop}
c_{1133}^{\overline{\rm TI}}=c_{3333}^{\overline{\rm TI}}-2c_{2323}^{\overline{\rm TI}}\,.
\end{equation}
Considering equation~(\ref{eq:deltaprop}) for two layers and assuming arithmetic average, we obtain
\begin{equation}\label{eq:long}
\left(\frac{1}{2}\left(\frac{a-2c}{a}+\frac{b-2d}{b}\right)\right)\,\left(\frac{1}{2}\left(\frac{1}{a}+\frac{1}{b}\right)\right)^{\!-1}+2\left(\frac{1}{2}\left(\frac{1}{c}+\frac{1}{d}\right)\right)^{\!-1}-\left(\frac{1}{2}\left(\frac{1}{a}+\frac{1}{b}\right)\right)^{\!-1}=0\,.
\end{equation}
After laborious algebraic computation, equation~(\ref{eq:long}) simplifies to
\begin{equation}\label{eq:short}
\left(c-d\right)\,\left(bc-ad\right)=0\,.
\end{equation}

For the Backus average, $\epsilon=0$ if and only if $c_{1111}^{\overline{\rm TI}}=c_{3333}^{\overline{\rm TI}}$, which for two layers is equal to
\begin{equation}\label{eq:long2}
\left(\frac{1}{2}\left(\frac{1}{a}+\frac{1}{b}\right)\right)^{\!-1}\left(\frac{1}{2}\left( \frac{a-2c}{a}+\frac{b-2d}{b}\right)\right)^{\!2}+\left(\frac{1}{2}\left(\frac{4(a-c)c}{a}+\frac{4(b-d)d}{b}\right)\right)=\left(\frac{1}{2}\left(\frac{1}{a}+\frac{1}{b}\right)\right)^{\!-1}\,.
\end{equation}
After laborious algebraic computation, equation~(\ref{eq:long2}) simplifies to
\begin{equation}\label{eq:short2}
\left(c-d\right)\,\left(c-d+b-a\right)=0\,.
\end{equation}
To receive $\delta=0$ and $\epsilon=0$, we need to solve equations~(\ref{eq:short}) and (\ref{eq:short2}).
Both equations are satisfied by $c=d$, which means that $c_{2323}$ is constant ($\gamma=0$) and the medium is isotropic.
If $c\neq d$, then both equations are satisfied by a system of equations,
\begin{equation*}
\begin{cases} 
bc=ad \\ c-d=b-a 
\end{cases}\,.
\end{equation*}
We obtain $b=-d$ and $a=-c$, which do not satisfy stability conditions, $b>\tfrac{4}{3}d$ and $a>\tfrac{4}{3}c$\,.
\end{proof}
\newpage
\section{Relation between $\gamma$ and $\gamma_{BV}$: Anisotropic medium}\label{sec:appb}
Let us consider a case of stronger inhomogeneity than that of Section~\ref{sec:quasi}. 
As shown in Table~\ref{tab:appa}, the differences among $c_{2323}$ within the stack of layers are greater.
\begin{table}[h]
    \centering
    \begin{tabular}{ ||p{2cm}||p{2cm}||}
\hline
 $c_{1111}$ & $c_{2323}$ \\
 \hline\hline
 10$x$&4\\
 \hline
 10&1\\
 \hline
 10$x$&4\\
 \hline
 10&1\\\hline
 10$x$&4\\
 \hline
 10&1\\\hline
 10$x$&4\\
 \hline
 10&1\\\hline
 10$x$&4\\
 \hline
 10&1\\
 \hline
\end{tabular}
    \caption{\small{Elasticity parameters for ten isotropic layers; factor $x$ controls the inhomogeneity of the stack.}}
    \label{tab:appa}
\end{table}
\begin{figure}[h]
    \centering
    \includegraphics[scale=0.5]{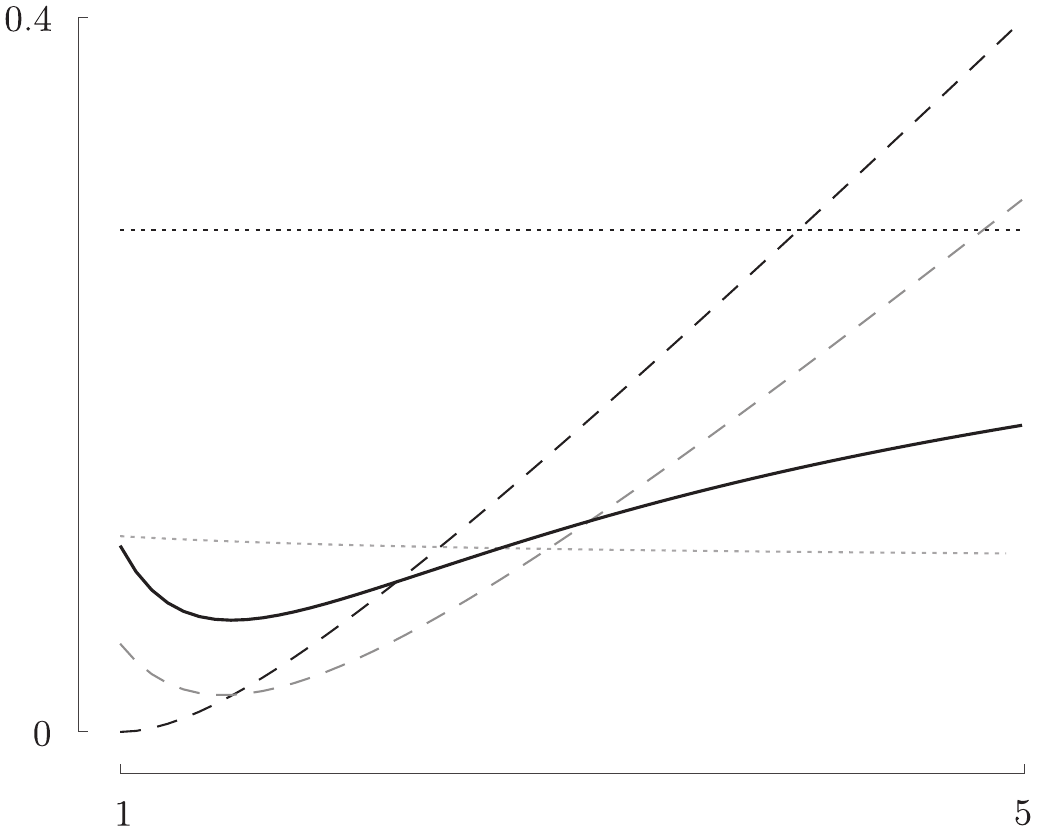}
    \caption{\small{Horizontal axis exhibits values of $x$. Values of $\mathscr{I}$ are shown by a dashed black line, $\mathscr{I}_{BV}$ by a dashed grey line, $\gamma$ by a dotted black line, $\gamma_{BV}$ by a dotted grey line and $\mathscr{N}$ by a solid black line.}}
    \label{fig:appa}
\end{figure}

As shown on Figure~\ref{fig:appa}, the relationship between $\gamma$ and $\gamma_{BV}$ is more sensitive to increasing values of $c_{1111}$, as compared to Figure~\ref{fig:appquasi}.
In other words, increasing inhomogeneity of $c_{1111}$ has a larger impact on the relationship between $\gamma$ and $\gamma_{BV}$ for strong inhomogeneity of $c_{2323}$, than for the weak one. 
Also, the relationship is larger than 2:1, due to stronger inhomogeneity of $c_{2323}$.

\end{appendix}
\end{document}